\documentclass{llncs} 
\usepackage{bm}
\usepackage{graphicx}
\usepackage{latexsym}
\usepackage{amsmath}
\usepackage{hyperref}
\usepackage{float}
\usepackage{algorithm}
\usepackage{algorithmic}
\usepackage{hyperref}
\title{\ Derandomization of Online Assignment Algorithms for Dynamic Graphs } 
\author{Ankur Sahai \thanks{asahai@cse.iitk.ac.in} }
\institute{Department of Computer Science, \\Indian Institute of Technology, Kanpur}
\date{ August 30, 2010 }
\begin{document}
\hyphenpenalty=100000
\setlength{\parindent}{10pt}
\setlength{\parskip}{1ex} 
\maketitle
\begin{abstract}

This paper analyzes different online algorithms for the problem of assigning weights to edges in a fully-connected bipartite graph that minimizes the overall cost while satisfying constraints. Edges in this graph may disappear and reappear over time. Performance of these algorithms is measured using simulations. This paper also attempts to derandomize the randomized online algorithm for this problem.
\end{abstract}

\section{Scope}
\label{sec:scope}
This paper aims to analyze online algorithms \ref{sec:online-algorithms} for dynamically evolving graphs\cite{rand-dyn-graphs,emp-dyn-graph-algos,multiflow-dyn-graph-algos}. The input consists of a bipartite graph $G = (V, E)$ with two types of nodes - consumers $C$ and producers $P$ (${V} = {C} \cup {P}$) and edges $E$ where $\{e_{ij} \in E: i \in C, j \in P\}$ and \emph{attribute arrays} associated with nodes $a_{v_i}(t)=[a_{v_{i1}}, a_{v_{i2}}, a_{v_{i3}},\cdots], v_i \in V$ and edges $a_{e_{ij}}(t)=[a_{e_{ij1}}, a_{e_{ij2}}, a_{e_{ij3}},\cdots], e_{ij} \in E$ whose values may change over time.

A sequence of online \emph{service requests}- $R = R(t) R(t+1) R(t+2)\cdots$ is received as input that consist of one or more consumer demands and edge failures. Service requests can contain more than one demands corresponding to consumers $R_k(t), k \in C$ - $R(t) = \cup_k R_k(t), k=|C|$ (corresponding to a multi-tape Turing machine). Here, t is the instance when service request $R(t)$ is received and is an increasing function of time. Demands act by either removing / adding edges or modifying edge attributes.

The objective is to minimize the overall cost of weight assignments such that it is not much worse than cost of optimal offline. 
\begin{equation}
\sum_{i \in C, j \in P} f(a_{e_{ij}}(t), R(t))*e_{ij}(t) \le\alpha*OPT(t),\: \forall t \in T 
\label{operation 1}
\end{equation} 

with constraints for consumers,
\begin{equation}
f(a_{e_{ij}}(t), R(t)) = f(a_{v_i}(t)), \forall i \in C
\label{operation 2}
\end{equation} 

and producers ,
\begin{equation}
f(a_{e_{ij}}(t), R(t)) = f(a_{v_j}(t)), \forall j \in P
\label{operation 3}
\end{equation} 

The dynamic nature of the edges is characterized by the following -
\begin{equation}
e_{ij}(t) = \left\{
\begin{array}{c l}
1 & if\mbox{ there is an edge between $i \in C$ and $j \in P$ at instance t} \\
0 & otherwise \mbox{ }
\end{array}
\right.
\label{operation 4}
\end{equation} 

This paper focuses on the optimal offline strategy and tries to find competitive online algorithms for this problem. This papers tries to use randomization to make the online algorithms more competitive. This papers also attempts to derandomize the randomized online algorithm.

\section{Problem Definition}
\label{sec:problem-definition}
This paper considers a subset of the problem specified in section \ref{sec:scope}. Given a complete bipartite graph $G = (V, E)$ where, ${V}$ is a finite set of nodes which consists of consumers ${c_i}, i \in C$ with indegree zero and producers ${p_j}, j \in P$ with outdegree zero such that, ${V} = {C} \cup {P}$ and edges $e_{ij} \in E$ where, $|E| = |V|^2$ with distances $d_{ij}$ between them. 

\begin{problem}
Online service requests $R=R_1R_2\cdots R_{n_1}, n_1 = |C|$ are received as input such that each service request has a unique demand $R_k, k \in C$ ($\cup R_{C \setminus k} = \emptyset, R_i \neq R_j , \forall i, j\in C$). These demands act by either increasing the edge weights $w_{ij}$s or removing an edge by setting $d_{ij} = \infty$. Edge weights $w_{ij}$s assigned to the edges cannot be reduced apart from the case where an edge goes down. In this case, weights are set to zero - $w_{ij} = 0 \forall e_{ij} = 0$.

Find an $\alpha$-competitive online algorithm for satisfying the service requests ${R}$ that minimizes the sum of weights:
\begin{equation}
\sum_{i \in C, j \in P} w_{ij}(t)*d_{ij}(t) *e_{ij}(t)\le\alpha*OPT(t), \forall t \in T
\label{operation 5}
\end{equation} 
where, c is a constant and OPT(t) is the output of the optimal offline algorithm at instance t. Such that, 

\begin{equation}
\sum_{i \in C, j \in P} w_{ij}(t) = R_i, \forall i \in C
\label{operation 6}
\end{equation}
Constraint in equation \ref{operation 6} guarantees that demands generated by the consumer until now are satisfied.

\begin{equation}
\sum_{i \in C, j \in P} w_{ij}(t) \le M_{j}, \forall j \in P
\label{operation 7}
\end{equation}
Constraint in equation \ref{operation 7} guarantees limited capacities for the producers.
\end{problem}

\begin{problem}
Consider a version of the problem \ref{sec:problem-definition} where edge distances $d_{ij}$ can change as specified by the service requests.
\end{problem}

\begin{problem}
Consider a version of the problem \ref{sec:problem-definition} where capacities associated with the producers $M_j$ can change as specified by the service request.
\end{problem}

\begin{problem}
Consider a version of the problem \ref{sec:problem-definition} where new consumers / producers can be added to the graph and some of the existing consumers / producers can go down.
\end{problem}

\section{Motivation}
The VMs running in a distributed system can be considered as the consumer and the data-centers as the producers of storage. The capacity of the producers can be considered as an attributes of the producer. The average time per I/O operation or latency between a VM and a data-center can be considered as the distance of the edge. And the sequence of I/O requests generated by the consumers can be considered at the service request in the problem \ref{sec:problem-definition}. And the edge failures are equivalent to the data-center or the VM being down.

The objective of this paper is to find a scheme for allocating these requests so that the overall cost of I/O operations at any instant is not more than $\alpha$ times worse than the optimal cost OPT that can be achieved when all the service requests are know at the beginning.

As more and more data moves to the cloud every day, it becomes important to analyze distributed resource scheduling schemes for better performance of VMs with respect to I/O read and write operations. To make the storage management transparent to the users of the Cloud platform it is important to have automated storage management schemes running on the cloud platform that make the best use of the available storage while guaranteeing good performance to the users. Aggregating the available storage across the distributed system into resource pools and distributing them prevents the storage from being wasted.

\section{Introduction}
One of the well known Distributed resource scheduling schemes is used in VMware's virtualization framework \cite{vmware-scale-storage} - Virtual Infrastructure using VirtualCenter - a centralized distributed system and, recently in VSphere - a cloud OS. Both these systems work by pooling the available storage into resource pools in a tree-like data structure. This paper \cite{survey-distributed-computing} analyzes the various distributed resource allocation techniques used in distributed systems.

We consider the theoretical aspects of this problem of allocating storage optimally to the VMs. Some of the points to be noted are as follows-
\begin{itemize}
\item[\ ]{{$\bullet$} Storage associated with a VM is a non decreasing function of time. }
\item[\ ]{{$\bullet$} Previously allocated storage cannot be moved to another data-center as this can be very time-consuming for large storage sizes.}
\item[\ ]{{$\bullet$} Some VMs may consume storage at a faster rate and can starve other VMs.}
\item[\ ]{{$\bullet$} To reduce the complexity we only consider a graph with fixed number of nodes; new VMs or data-centers are not added dynamically.}
\end{itemize}

Similar problems involving min-flow \cite{min-flow-related}, online matching \cite{online-matching,online-weighted-bip-match}, dynamic assignment \cite{dynamic-assignment}, LP techniques \cite{interior-point-methods,lp-onlline}, bipartite network flow \cite{bip-network-flow} and combinatorial optimization \cite{distributed-comb-opt}, distributed resource allocation \cite{constraint-driven-scheduling,chromatic-sums,job-assign-scalable} have been studied in earlier. Hungarian algorithm\cite{hungarian-algo} is one of the earliest known method for solving the assignment problem that uses the primal-dual method. It is to be noted that the problem studied in this paper does not require a matching.

Fairness of resource allocation \cite{fairness-alloc,balanced-alloc} and dynamic load balancing \cite{online-load-balancing,related-load-balancing,hierearchical-load-balancing} issues related to these problems have also been studied before. \cite{power-of-randomization-online} studies how randomization can be used to improve the competitiveness of online algorithms. \cite{cloud-computing} throws light on how such schemes could be adapted to large scale cloud-computing platforms.

\section{Offline Algorithms}
Optimal offline algorithm for this problem is a Linear Program. It seems trivial at first to iterate through the demands and allocate weights corresponding to demands, on the edges with least distance $d_{ij}$ that are connected to producers with available capacity. However, the order in which demands are considered will affect the cost of output. The optimal algorithm has to look at all the demands simultaneously and then look at all the available edges where the demands can be allocated.

Due to this, one has to consider all possible assignments of demands to edges in order to find the least cost assignment. This paper uses LP for solving the offline version due to ready availability of LP code that is used for simulations in section \ref{sec:simulations}.
\subsection{LP}
\label{sec:offline-LP}
The LP formulation for this problem is as follows -

Objective function:
\begin{equation}
Minimize: \sum_{i\in C, j \in P} d_{ij}*w_{ij}, \; w_{ij} >= 0, \: d_{ij} > 0
\label{operation 8}
\end{equation}

As the demands are non-negative the weights are also non-negative. Edges that fail \ref{operation 4} have their corresponding $d_{ij}$'s set to infinity so that, they are not selected.

Constraints:
\begin{equation}
\sum_{i\in C, j \in P} w_{ij} \ge R_i, \forall i \in C
\label{operation 9}
\end{equation}

Equation in \ref{operation 9} suggests the total demand of a consumer $i \in C$ should be met.
\begin{equation}
\sum_{i\in C, j \in P} w_{ij} \le M_{j} \implies -\sum_{i\in C, j \in P} w_{ij} \ge -M_{j}, \; \forall j \in P
\label{operation 10}
\end{equation}

Equation in \ref{operation 10} suggests, the capacities of the producers cannot not be exceeded.
\begin{lemma}
LP formulation in \ref{sec:offline-LP} produces a valid assignment of weights $w_{ij}$ on edges $e_{ij}$ corresponding to the demands R.
\end{lemma}
\begin{proof}
Equation \ref{operation 9} gurantees that the total demand generated by consumer $i \in C$ is satisfied. Equation \ref{operation 10} ensures that the capacities of producers $j \in P$ are not exceeded. By definition \ref{sec:problem-definition} this is a valid assignment of weight on edges.\qed
\end{proof}

\begin{lemma}
LP formulation in \ref{sec:offline-LP} produces the optimal assignment of weights $w_{ij}$ on edges $e_{ij}$ corresponding to the demands in R.
\end{lemma}
\begin{proof}
Lemma 1 ensures that this LP produces a valid solution. Since the objective \ref{operation 8} is a minimization function and fractional weights are allowed, it follows that the solution produced by LP is the optimal solution.\qed
\end{proof}
For edge failures, the LP formulation \ref{sec:offline-LP} has to be modified by removing the failed edges, adding constraints for the current weight assignments and adding the demands. 
\section{Online Algorithms}
\label{sec:online-algorithms}
Online algorithms are used for solving problems where the entire input is not known and partial decisions have to be made at each step. The input is received incrementally. Competitive ratio is used to measure the performance of online algorithm as compared to the optimal offline algorithm that knows the entire input.

A $\alpha$-competitive online algorithm ALG is defined as follows with respect to an optimum offline algorithm OPT -
\begin{equation}
cost(ALG(I)) \le\alpha*cost(OPT(I)) + \beta
\label{operation 14}
\end{equation}

In the definition of online algorithm in equation \ref{operation 14}, $\alpha$ is called the \emph{competitive ratio} and $\beta$ can be considered as the \emph{startup cost} of the algorithm.

One of the widely used online algorithm is for the k-server problem \cite{k-server-offline,k-server-lower-bounds,k-server-competetive,k-server-decomposition,k-server-randomized,k-server-survey} where k-servers have to service n clients and the order of the service requests from the clients is not known at the beginning. The objective of this problem is to find the shortest path to serve the clients.
\subsection{Greedy Algorithm}
\label{sec:online-greedy}
This algorithm looks for the best available edge - $min_{d_{ij}}$. This leads to a myopic behavior as the algorithm always looks for local minima. For example, consider a graph consisting of two consumers $c_1$ and $c_2$ with demands $R_1$ and $R_2$ such that, $R_1 < R_2$ and two producer $p_1$ and $p_2$ with capacities $M_1$ and $M_2$ respectively such that $M_1 > M_2$ and $R_1 = M_1$. Let the edges be $e_{11}$, $e_{12}$, $e_{21}$ and $e_{22}$ such that $d_{11} = d{21} = x $ and $d_{12} = d{22} = x+5$.

Demand $R_1$ from consumer $c_1$ arrives first and is allocated on the cheapest edge $e_{11}$. Since, $R_1 = M_1$ the demand $ d_2$ from consumer $c_2$ that arrives next is allocated on the edge $e_{22}$. The total cost is $x*R_1 + (x+5)*R_2$ = $x*(R_1 + R_2) + 5*R2$. The cost would have been $(x+5)*R_1+x*R_2$ = $x*(R_1 + R_2) + 5*R_1$ which is lesser than the cost of output produced by greedy since $R2 < R1$.

It is observed that,
\begin{equation}
\begin{split}
w_{ij} \alpha (1/d_{ij}), \forall i \in C, j \in P \\
\implies w_{ij} = k*(1/d_{ij}), \forall i \in C, j \in P
\end{split}
\label{operation 11}
\end{equation}

Substituting equation \ref{operation 11} in \ref{operation 10} we get -

\begin{equation}
\begin{split}
\sum_{i \in C, j \in P} k*(1/d_{ij}) \le M_{j}, \forall j \in P \\
\implies k \le [1/\sum_{i \in C, j \in P} (d_{ij})]*M_j, \forall j \in P
\end {split}
\label{operation 13}
\end{equation}

This paper aims to analyze an online algorithm that is more far-sighted. The Randomized-Greedy algorithm in the next section also considers the non-optimal edges to prevent getting stuck in local minima. 
\subsection{Randomized Greedy Algorithm}
\label{sec:online-randomized-greedy}
\begin{algorithm}
\caption{Randomized-Greedy}
\label{alg1}
\begin{algorithmic}
\STATE $S \gets FindTopAvail(k)$
\STATE $e_G \gets Greedy(S)$
\STATE $numIterations \gets 0$
\STATE $maxIterations \gets n$
\STATE $e_G \gets Greedy(S)$
\WHILE{$numIterations \leq maxIterations$}
\STATE $numIeratations \gets numIterations + 1$
\STATE $e_R \gets Random(S)$
\IF{$ d(e_R) / d(e_G) \le \beta $}
\RETURN $e_R$
\ENDIF
\ENDWHILE
\RETURN $e_G$
\end{algorithmic}
\end{algorithm}

This algorithm intends to be more far-sighted than the greedy algorithm. For this, it maintains a set of top k cheapest available edges, sorted by $d_{ij}$. At each step we consider and edge at random from this set and compare its cost of the edge selected by Greedy algorithm. If an edge meeting the condition $d(e_R) / d(e_G) \le \beta$ is found then we select this edge. This procedure is repeated {\emph n} number of times. If no such edge is found within n iterations then we simply return the edge selected by Greedy.

In this algorithm $\beta$ is called the sub-optimal penalty. When it is set to 1 then it becomes a greedy algorithm. As $\beta$ increases the algorithm is allowed to select higher cost edges. The aim is to prevent the algorithm from getting stuck in a local optima. As it is seen in the simulation results \ref{cost-rand}, this algorithm performs better when the distribution of service requests is not uniform.
\subsubsection{Derandomization of Randomized-Greedy algorithm}
\label{sec:derandomization}
This papers aims to measure the performance of online algorithms using pairwise-independent random numbers as input in place of the real-world data. Recursive n-gram hashing is used to generate consumer demands and producer capacities.
The randomized greedy algorithm is derandomized by running it multiple times and the output with the least cost is selected among the available outputs. Using this method requires higher processing power although, it improves the performance of Randomized algorithm.
\section{Simulations}
\label{sec:simulations}
This paper simulates the online assignment problem using C++ - \(GCC - 4.3.0 20080428\), LP solver \(lp\_solve v.5.1.1.3\) and shell scripts on Linux (Red Hat 4.3.0-8) platform. Shell script is used to generate input using built-in random number generator for edge distances, producer capacities, consumer demands, edge failures. The online greedy and randomized algorithms \ref{alg1} are implemented in C++. The binary file containing these algorithms takes three parameters: type of algorithm, input file and output file.

The outermost shell script is invoked as -
\begin{verbatim}
evaluateDRSUsers1.sh $users $demands $resources $capacities $failures
\end{verbatim}

This script generates an input file for the online algorithms -
\begin{verbatim}
no of consumers: 2
no of producers: 2
edge distances
47
17
11
2
producer capacities
26
839
consumer demands
97
78
Number of edge failures: 1
1 <- demand #
3 <- edge #
\end{verbatim}
and for the offline LP solver -
\begin{verbatim}
min: 47x1+17x2+11x3+2x4; <- objective function
x1+x3<=26; <- producer constraint 
x2+x4<=839;
x1+x2=97;  <- consumer constraint
x3+x4=78;
\end{verbatim}

This script then invokes the LP solver -
\begin{verbatim}
./lp_solve model.lp >> drsUsersOut.txt
\end{verbatim}

and the binary file containing the greedy and randomized algorithms -
\begin{verbatim}
./onlineDRSAlgo1.o greedy onlineIn.txt drsUsersOut.txt
./onlineDRSAlgo1.o randomized onlineIn.txt drsUsersOut.txt
\end{verbatim}

Main function in C++ for solving the online version of the problem -
\begin{verbatim}
int main(int argc, char* argv[]) {
   GetAlgo(argv[1]);
   ReadInputFile(argv[2]);
   GetNumProducers();
   GetNumConsumers();
   GetEdgeDist();
   InitCEWeights();
   InitEWeights(); 
   GetConsumerDemands();
   GetProducerCapacities();
   InitEdgeAvail();
   GetEdgeFailures();
   AllocateDemands();
   WriteOutputFile(argv[3]);
   return 0;
}
\end{verbatim}

\begin{figure}
\begin{center}
\includegraphics[width=100 mm]{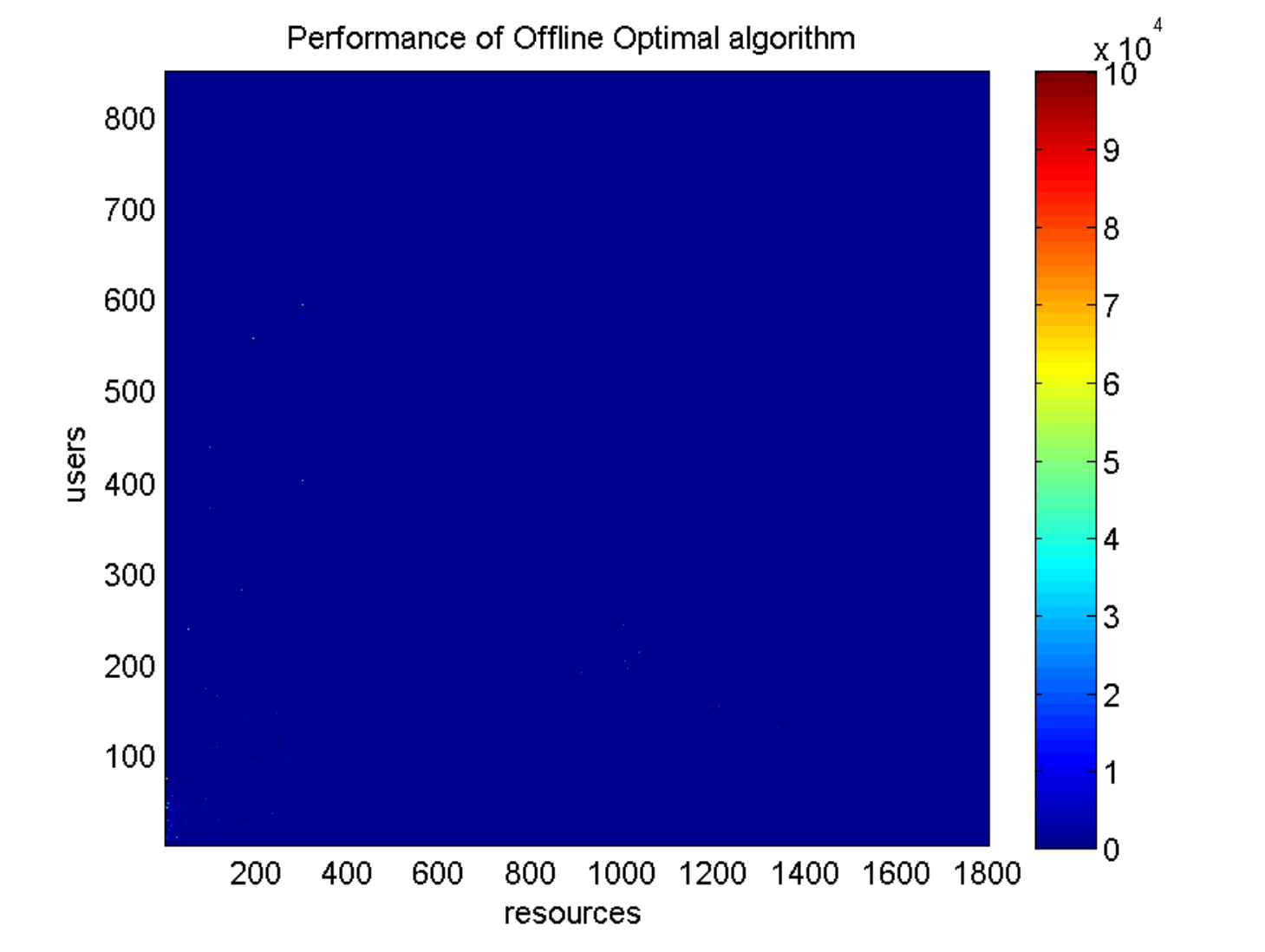}
\end{center}
\caption{Cost vs Offline Optimal}
\label{cost-opt}
\end{figure}
\begin{figure}
\begin{center}
\includegraphics[width=100 mm]{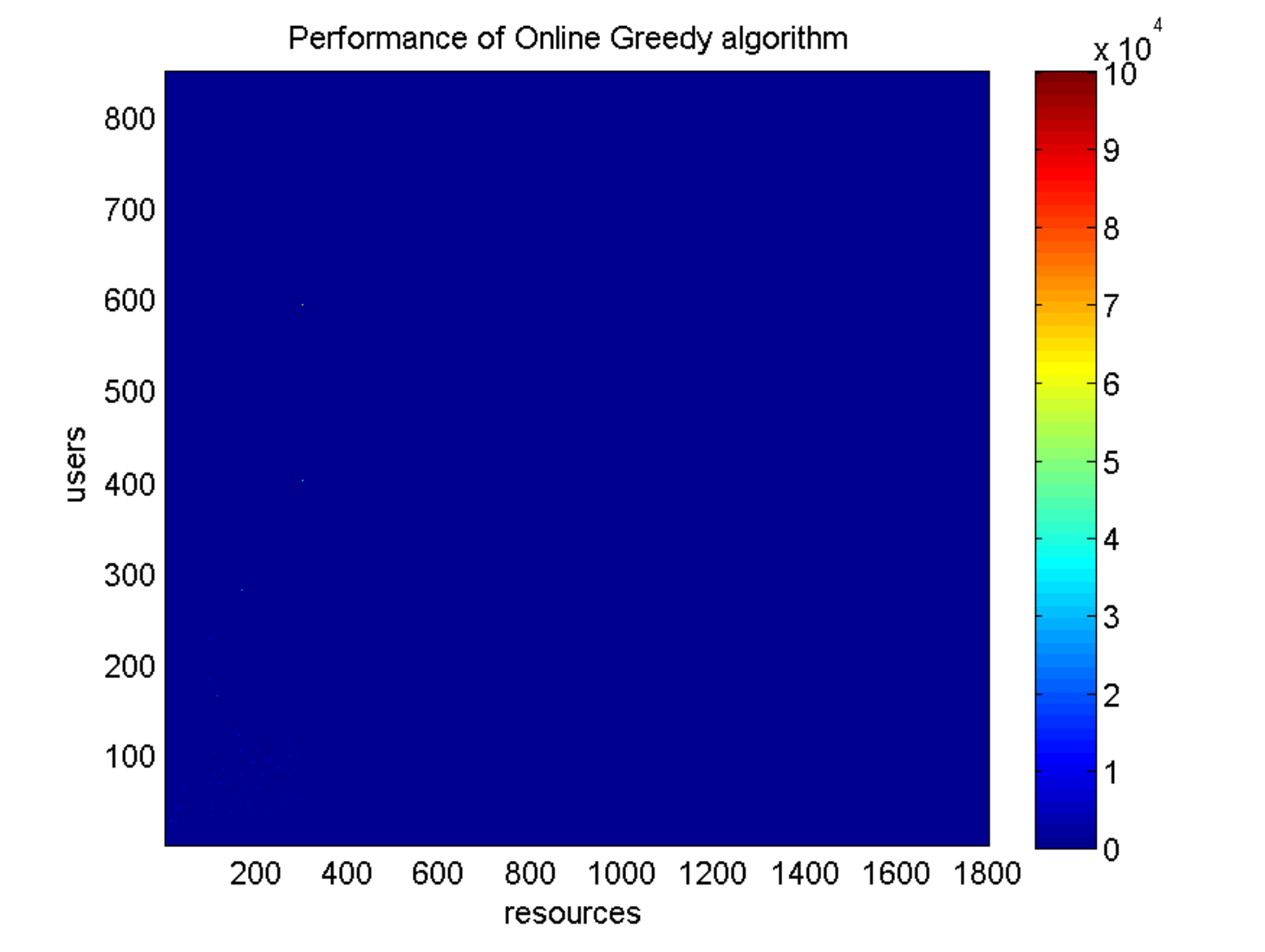}
\end{center}
\caption{Cost vs Online Greedy}
\label{cost-greedy}
\end{figure}
\begin{figure}[H]
\begin{center}
\includegraphics[width=100 mm]{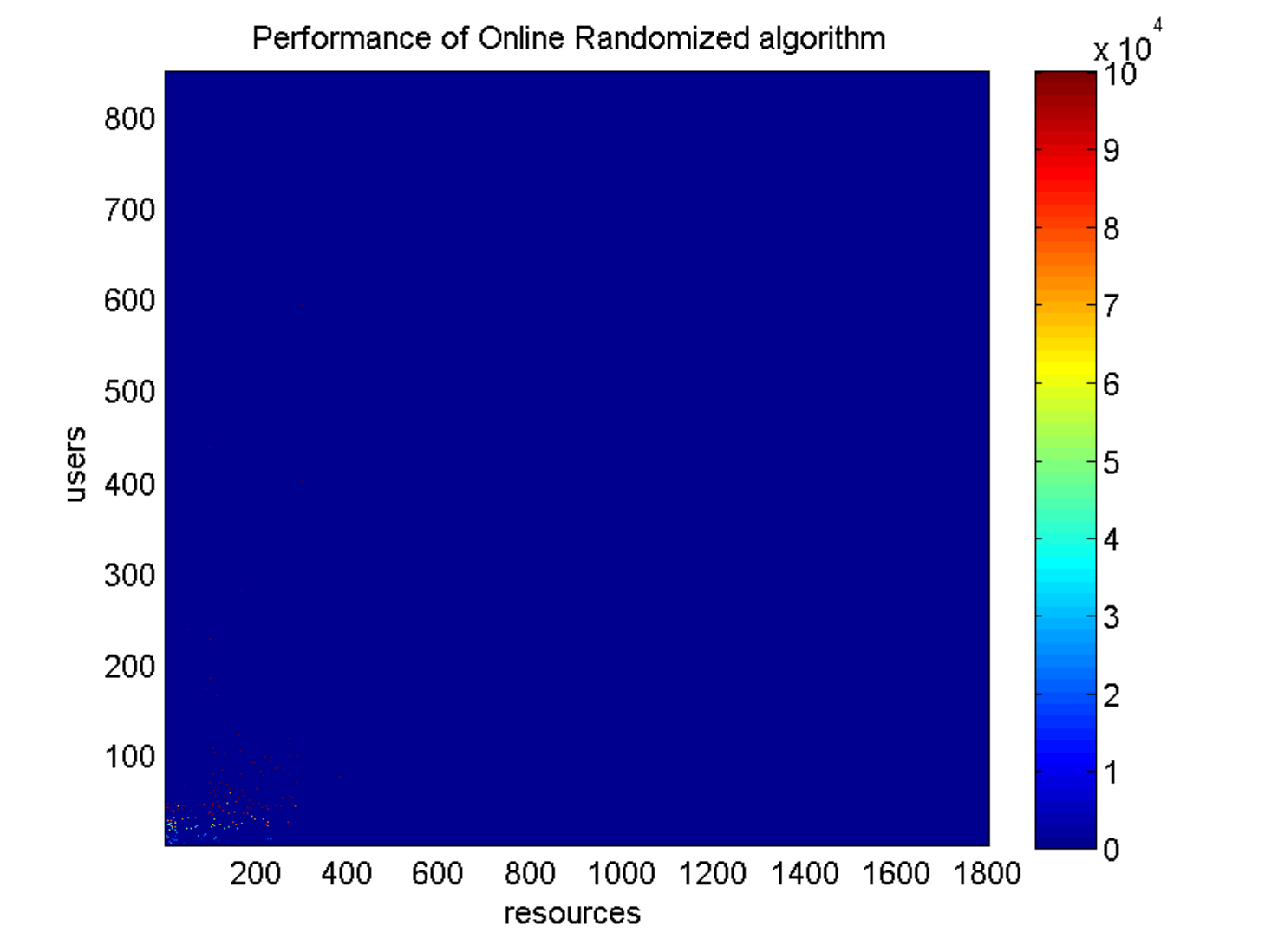}
\end{center}
\caption{Cost vs Online Randomized}
\label{cost-rand}
\end{figure}

The weights of edges are initialized to zero. The availabilities of edges are initialized to the capacity of the producers connected. For greedy algorithm, the program chooses the available edge with the least cost. For the randomized algorithm, the program first sorts the available edges in increasing order. Then picks an edge suggested by the srand() and rand() function such that the cost of this edge is less than $\beta$ times the best available edge.

The weight of the edge returned - $weight(e)$, is increased by the value of demand and the availability $avail(e)$ is decreased by the same value. If there are any edge failures at the current demand the program generates a new internal demand and sets $avail(e) = weight(e) = 0$ for the edge that failed. The program writes the total cost along with the $weight(e), \forall e \in E$ to the output file. STL map data-structures with double data-types are used for edge and node properties. STL 2-D vectors with double data-types are used for tracking the weights allocated on the edges for each consumer. 

For finding the optimal offline solution in the case of edge failures additional constraints are added to LP at each instance of edge failure, such that the currently allocated weights are not disturbed. This is done for each of the edge failures with the demand of the consumers at each stage being equal to their sum of demands until the edge failed.


\subsection{Analysis of Simulation Results}
\label{sec:simulation-results-analysis}
In \emph{Fig. \ref{cost-opt}} as the number of users increases with the same number of resources, the cost of the optimal offline also increases. The number of incoming edges to any producer increases as the consumers increase. This increases the competition for the resources available at the producers. The optimal offline algorithm has to select costlier edges due to limited availability of edges, as the number of demands increase and this increases the overall cost. This is inline with Fig 4 in \cite{vmware-scale-storage} where it is referred to the latency (which is the overall time required to complete I/O requests). 

Greedy algorithm \emph{Fig. \ref{cost-greedy}} closely follows the optimal offline algorithm at each step of the input. greedy which produces output that has a cost comparable to the optimal solution and does not deviate much apart from the cases where higher demands arrive later as described in paragraph 2 in section \ref{sec:online-greedy}.

Randomized algorithm does better than Greedy algorithm in cases where higher demands arrive later in the sequence. In this case, greedy algorithm fails badly as it does not have any cheaper edges left towards the end and has to select costlier edges for higher demands. However, Randomized algorithm\emph{Fig. \ref{cost-rand}} does not cope well for graphs with higher number of nodes. Randomized algorithm pays the penalty of selecting a suboptimal edge to get an overall cost improvement; however due to large size of the demands it is not able to recover from this penalty and this causes a cascading effect. In \emph{Fig.3} the randomized online algorithm shows the highest standard deviation because of the random nature of the selections.

\section{Conclusion}
\label{sec:conclusion}
This paper concludes that it is possible to produce a competitive online algorithm for this problem using randomization. The competitiveness of this algorithm suffers for higher number of nodes and edges in the graph. This paper believes that it is possible to produce a competitive version of this algorithm by setting the suboptimal penalty to zero at the beginning of the algorithm and iteratively increasing it depending on the effect on the cost of the output.

\section{Appendix}
\label{sec:appendix}
The results found this paper are based on the simulation experiments. It will be interesting to measure the performance of these online algorithms for real-world I/O demands. And, how it scales with the size of input.

Work on solving other versions of this problem where the distances $d_{ij}$s of edges may change over time is currently in progress.

\section{Glossary}
\label{sec:glossary}
\begin{definition} \emph {Dynamic graphs:} Graphs that have a finite set of nodes and edges but where, the edges or nodes may become unavailable and available over time. \end{definition}
\begin{definition} \emph {Derandomization:} Process of removing randomness from an algorithm to make it more deterministic. \end{definition}
\begin{definition} \emph {Assignment problem:} Given a bipartite graph with tasks and resource on either side and a cost associated with allocating a task to a resource, the assignment problem is used to find the optimal allocation of the resources to the tasks (to minimize a Objective function) \end{definition}
\begin{definition} \emph{Linear programming:} A method used to solve large scale optimization problems with set constraints and objective function (minimize or maximize a quantity) \end{definition}
\begin{definition} \emph{Combinatorial optimization:} A certain class of optimization problems that involves minimization of a Objective function where there are multiple choices available at each step. \end{definition}

\end{document}